\documentclass[12pt,onecolumn,twoside]{IEEEtran}

\usepackage{setspace}
\usepackage{color,soul}
\usepackage{amsfonts}
\usepackage{dsfont}
\usepackage[cmex10]{amsmath}
\usepackage{amssymb}

\usepackage{algorithm}
\usepackage{algpseudocode}
\algnewcommand{\algorithmicgoto}{\textbf{go to}}%
\algnewcommand{\Goto}[1]{\algorithmicgoto~\ref{#1}}%
\newtheorem{theorem}{Theorem}
\newtheorem{proof}{Proof}

\newtheorem{remark}{Remark}
\newtheorem{definition}{Definition}

\usepackage{graphicx}
\graphicspath{ {./} }
\DeclareGraphicsExtensions{.pdf}

\begin{document}

\title{Continuous Time Quantum Consensus \& Quantum Synchronisation }

\author{Saber Jafarizadeh,~\IEEEmembership{Member,~IEEE}, \thanks{e-mail: saber.jafarizadeh@sydney.edu.au }   }

\maketitle


\bibliographystyle{ieeetran}


\begin{abstract}

Distributed consensus algorithm over networks of quantum systems has been the focus of recent studies in the context of quantum computing and distributed control.
Most of the progress in this category have been on the convergence conditions and optimizing the convergence rate of the algorithm, for quantum networks with undirected underlying topology.
This paper aims to address the extension of this problem over quantum networks with directed underlying graphs.
In doing so, the convergence to two different stable states namely, consensus and synchronous states have been studied.
Based on the intertwining relation between the eigenvalues, it is shown that for determining the convergence rate to the consensus state, all induced graphs should be considered while for the synchronous state only the underlying graph suffices.
Furthermore, it is illustrated that for the range of weights that the Aldous’s conjecture holds true, the convergence rate to both states are equal.
Using the Pareto region for convergence rates of the algorithm, the global and Pareto optimal points for several topologies have been provided.

%

\end{abstract}


\section{Introduction}

In the context of distributed control, distributed consensus  algorithms are employed as the building block for implementing other distributed algorithms  which rely on the individual decision of agents and the local communication among them \cite{Xiao04Boyd,SaberThesis2015}.
The extension of this class of algorithms to the quantum domain has been addressed in \cite{PetersenRef15} where four different generalizations of classical consensus states have been proposed.
In \cite{Petersen2015IEEETranAutControl,Petersen2015ACCPartI,Petersen2015ACCPartII} the necessary and sufficient conditions for asymptotic convergence of the quantum consensus algorithm is studied.
Optimizing the convergence rate of the algorithm to the consensus state has been addressed in \cite{SaberQConsensusContinuous,SaberQConsensusDiscrete}.
The majority of the analysis regarding the convergence rate of the algorithm has been focused on quantum networks with an undirected underlying graph.
%
%
%
In this paper, we aim to study the convergence rate of the distributed consensus algorithm over a network of qudit systems with general (i.e. either directed or undirected) underlying topologies.
The convergence rates to two different states of the network of quantum systems have been studied.
These states are consensus and synchronous states.
Consensus state is the symmetric state which is invariant to all permutations \cite{PetersenRef15}.
%
Synchronous state is the state where the reduced states of the quantum network are driven to a common trajectory \cite{Petersen2015ACCPartI,Petersen2015ACCPartII}.
Employing the intertwining relation \cite{SaberQConsensusContinuous} between the eigenvalues of the Laplacian matrices of the induced graphs, we have shown that the convergence rate to the consensus state is obtained from the spectrum of all induced graphs combined.
On the contrary, the convergence rate to the synchronous state is dictated by only the spectrum of the underlying graph of the network and therefore it is independent of the dimension of the Hilbert space $(d)$.
By establishing the relation between the convergence rates to consensus and synchronous states, we have shown that both convergence rates are equal and independent of $d$ if the Aldous' conjecture holds true for all induced graphs of the networks.
%
%
Furthermore, we have proved that the synchronous state is reachable for any permutation-invariant system Hamiltonian $(H_0)$, while for the algorithm to converge to the consensus state, either the system Hamiltonian should be zero (i.e. $H_0 = 0$) or the analysis should be limited to the interaction picture.
For different network topologies, by plotting the Pareto region of the convergence rates to the consensus and synchronous states, we have studied the Pareto optimal points and the global optimal points regarding both convergence rates.

The rest of the paper is organized as follows.
Preliminaries on graph theory are provided in section \ref{Preliminaries}.
Section \ref{sec:EvolutionQNet} explains the evolution of the quantum network.
In section \ref{Sec:Simulations}, optimization of the convergence rates of the algorithm have been addressed over different topologies and Section \ref{sec:Conclusion} concludes the paper.

\section{Preliminaries}
\label{Preliminaries}

In this section, we present the fundamental concepts 
on graph theory, Cayley and Schreier Coset Graphs.

\subsection{Graph Theory } 
\label{sec:GraphPreliminaries}

A directed graph (digraph) is defined as $\mathcal{G}= \{ \mathcal{V}, \mathcal{E} \}$ with $\mathcal{V} = \{1, \ldots, N\}$ as the set of vertices and $\mathcal{E}$ as the set of edges.
Each edge $\{i,j\} \in \mathcal{E}$ is an ordered pair of distinct vertices, denoting an directed edge from vertex $i$ to vertex $j$.
%
%
Throughout this paper, we consider directed simple graphs with no self-loops and at most one edge between any two different vertices.
A weighted graph is a graph where a weight is associated with every edge according to a proper map $W: \mathcal{E} \rightarrow \mathbb{R}$, such that if $\{i,j\} \in \mathcal{E}$, then $W(\{i,j\})= \boldsymbol{w}_{ij}$; otherwise $W(\{i,j\}) = 0$.
The edge structure of the weighted graph $\mathcal{G}$ is described through its weighted adjacency matrix $(A_{\mathcal{G}})$.
The weighted adjacency matrix $A_{\mathcal{G}}$ is a $N \times N$ matrix with $\{i,j\}$-th entry $(A_{\mathcal{G}}(i,j))$ defined as below
\begin{equation}
    \nonumber
    \begin{gathered}
        \nonumber  A_\mathcal{G}(i,j) =
	       \begin{cases}
                \boldsymbol{w}_{ij} \quad \text{if} \quad \{i,j\} \in \mathcal{E}, \\
                0 \quad \text{Otherwise.}
            \end{cases}
     \end{gathered}
\end{equation}
The indegree (outdegree) of a vertex $i$ is the sum of the weights on the edges heading in to (heading out of) vertex $i$.
%
%
A directed path (dipath) in a digraph is a sequence of vertices with directed edges pointing from each vertex to its successor in the sequence. A simple dipath is the one with no repeated vertices in the sequence.
%
%
A directed graph is called strongly connected if there is a dipath between any pair of vertices in the graph.
The weighted Laplacian matrix of graph $\mathcal{G}$ is defined as below,
\begin{equation}
    \nonumber
    \begin{gathered}
        \nonumber  L_{\mathcal{G}}(i,j) =
	       \begin{cases}
                D_{\mathcal{G}}(i,i) \quad \text{if} \quad i = j, \\
                -A_{\mathcal{G}}(i,j) \quad \text{if} \quad i \neq j,
            \end{cases}
     \end{gathered}
\end{equation}
where $D_{\mathcal{G}}(i,i)$ is the indegree of the $i$-th vertex.
This definition of the weighted Laplacian matrix can be expressed in matrix form as $L_{\mathcal{G}} = D_{\mathcal{G}} - A_{\mathcal{G}}$, where $D_{\mathcal{G}}$ and $A_{\mathcal{G}}$ are the indegree and the adjacency matrices of the graph $\mathcal{G}$.
The Laplacian matrix of a directed graph is not necessarily symmetric and
its eigenvalues can have imaginary parts.
Defining $\bf{1}$ and $\bf{0}$ as vectors of length $N$ with all elements equal to one and zero, respectively, for the Laplacian matrix we have $L_{\mathcal{G}} \times \bf{1} = \bf{0}$.
In directed graphs, the  eigenvalues of the the associated Laplacian can be arranged in non-decreasing order as below,
\begin{equation}
    \nonumber
    \begin{gathered}
        0 = \lambda_{1} (L_{\mathcal{G}})  \leq Re\left(\lambda_{2} (L_{\mathcal{G}}) \right) \leq \cdots \leq Re\left( \lambda_{N} (L_{\mathcal{G}}) \right).
     \end{gathered}
\end{equation}
$Re\left(\lambda_{i} (L_{\mathcal{G}}) \right)$ denotes the real part of the $i$-th eigenvalue of the weighted Laplacian matrix of the digraph.
The digraph $\mathcal{G}$ is said to be strongly connected \cite{Chung1997} iff $Re\left(\lambda_{2} (L_{\mathcal{G}}) \right) > 0$.

\subsection{Cayley Graph \& Schreier Coset Graph}

Let $\mathcal{H}$ be a group and let $\mathcal{S} \subseteq \mathcal{H}$.
The Cayley graph of $\mathcal{H}$ generated by $\mathcal{S}$ (referred to as the generator set $\mathcal{S}$), denoted by $Cay(\mathcal{H}, \mathcal{S})$, is the directed graph $\mathcal{G} = (\mathcal{V}, \mathcal{E})$ where $\mathcal{V} = \mathcal{H}$ and $\mathcal{E} = \{(x, xs) | x \in \mathcal{H}, s \in \mathcal{S}\}$.
If $\mathcal{S} = \mathcal{S}^{-1}$ (i.e., $\mathcal{S}$ is closed under inverse), then $Cay(\mathcal{H}, \mathcal{S})$ is an undirected graph.
If $\mathcal{H}$ acts transitively on a finite set $\Omega$, we may form a graph with vertex set $\mathcal{V} = \Omega$ and edge set $\mathcal{E} = \{ (\nu, \nu s) | \nu \in \Omega, s \in \mathcal{S} \}$. Similarly, if $\mathcal{Q}$ is a subgroup  in $\mathcal{H}$, we may form a graph whose vertices are the right cosets of $\mathcal{Q}$ , denoted $(\mathcal{H}:\mathcal{Q})$ and whose edges are of the form $\mathcal{E} = \{(\mathcal{Q}h, \mathcal{Q}hs) |\mathcal{Q}h\in (\mathcal{H}:\mathcal{Q}), s \in \mathcal{S}\}$.
These two graphs are the same when $\Omega$ is the coset space $(\mathcal{H}:\mathcal{Q})$, or when $\mathcal{Q}$ is the stabilizer of a point of $\Omega$ and  is called the Schreier coset graph $Sch(\mathcal{H}, \mathcal{S}, \mathcal{Q})$.

\section{Evolution of the Quantum Network}
\label{sec:EvolutionQNet}

Considering a quantum network as a composite (or multipartite) quantum system with $N$ qudits,
and assuming $\mathcal{H}$ as the d-dimensional Hilbert space over $\mathbb{C}$, then the state space of the quantum network is within the Hilbert space $\mathcal{H}^{\otimes N} = \mathcal{H} \otimes \ldots \otimes \mathcal{H}$.
The state of the quantum system is described by its density matrix $(\boldsymbol{\rho})$ (a positive Hermitian matrix with trace one).
The network is associated with an underlying graph $\mathcal{G}=\{ \mathcal{V}, \mathcal{E} \}$, where $\mathcal{V}=\{1,\ldots, N\}$ is the set of indices for the $N$ qudits, and each element in $\mathcal{E}$ is 
an ordered pair of two distinct qudits, denoted as $\{j,k\} \in \mathcal{E}$ with $j,k \in \mathcal{V}$.
Permutation group $S_{N}$ acts in a natural way on $\mathcal{V}$ by mapping $\mathcal{V}$ onto itself.
For each permutation $\pi \in S_{N}$
we associate unitary operator $U_{\pi}$ over $\mathcal{H}^{\otimes N}$, as below
\begin{equation}
    \nonumber
    \begin{gathered}
        U_{\pi} ( Q_{1} \otimes \cdots \otimes Q_{N} ) = Q_{\pi(1)} \otimes \cdots \otimes Q_{\pi(N)},
     \end{gathered}
\end{equation}
where $Q_{i}$ is an operator in $\mathcal{H}$ for all $i = 1, \ldots, N$.

Employing the quantum gossip interaction introduced in \cite{PetersenRef15}, the evolution of the quantum network can be described by the following master equation
%
%
\begin{equation}
    \label{eq:Lindblad2}
    \begin{gathered}
        \frac{d\boldsymbol{\rho}}{dt} = - \frac{i}{\hbar} [ H_0 , \boldsymbol{\rho} ]  +  \sum\nolimits_{ \pi \in \textit{B}} { w_{\pi} \left(  U_{\pi} \times \boldsymbol{\rho} \times U_{\pi}^{\dagger} - \boldsymbol{\rho}  \right) }
     \end{gathered}
\end{equation}
where $\textit{B}$ is a subset of the permutation group $S_N$,
$H_0$ is the (time-independent) system Hamiltonian, 
$i$ is the imaginary unit, $\hbar$ is the reduced Planck constant and 
where $w_{\pi}$ is a positive 
time-invariant
weight corresponding to the permutation $ \pi$.
These weights form the distribution of limited amount of weight up to $D$, among edges of the underlying graph, i.e.
\begin{equation}
    \label{eq:Lindbladconstraint}
    \begin{gathered}
        \sum\nolimits_{ \pi \in \textit{B} } {\textsf{l}_{\pi}w_{\pi}} \leq D,
     \end{gathered}
\end{equation}
where $\textsf{l}_{\pi}$ is the 
sum of the cycle length of cycles
appearing in $\pi$ except for trivial one-cycles.
%
%
%
%
The generator set $\textit{B}$ should be selected in a way that the underlying graph corresponding to $G_{\textit{B}}$ (the group generated by $\textit{B}$ which is a subset of $S_N$) is connected.


\subsection{Consensus \& Synchronous States}

In \cite{PetersenRef15}, four different consensus states generalized to the quantum domain are exploited.
Based on these schemes, three different possible consensus states can be defined which are reachable by quantum consensus algorithm.
%

\begin{definition} {Observable-Expectation Consensus}
\label{AvgStateDef}

The observable-expectation consensus $(\boldsymbol{\rho}^{*})$ is defined as the state where for any observable $\sigma$ the following holds,
%
\begin{equation}
    \label{eq:AverageState}
    \begin{gathered}
        tr(\boldsymbol{\rho}^{*}(t)\sigma_l)  =  tr(\boldsymbol{\rho}^{*}(t)\sigma_k)\quad \text{for}\quad k,l=1,\cdots,N,
    \end{gathered}
\end{equation}
where $\sigma_k= I\otimes I \otimes \cdots \overbrace{\sigma}^k \otimes \cdots I)$.
\end{definition}

\begin{definition} {Synchronous State}
\label{SynchDef}

The synchronous state $(\boldsymbol{\rho}^{*})$ is defined as the state 
where the following holds
%
\begin{equation}
    \label{eq:SynchronousState}
    \begin{gathered}
        \bar{\boldsymbol{\rho}}^{*}_{1}  =  \bar{\boldsymbol{\rho}}^{*}_{2}  =  \cdots  =  \bar{\boldsymbol{\rho}}^{*}_{N},
    \end{gathered}
\end{equation}
where $\bar{\boldsymbol{\rho}}_{k}$ is the reduced state of the subsystem $k$ for an overall system state $\boldsymbol{\rho}$ i.e. $\bar{\boldsymbol{\rho}}_{k}  =  tr_{(\otimes_{j\neq k}H_{j})} \left( \boldsymbol{\rho} \right)$.
In \cite{PetersenRef15}, equation $(\ref{eq:SynchronousState})$ is defined as the reduced state consensus.

\end{definition}
%
%
%
%
%
In \cite{PetersenRef15} it is shown that the observable-expectation consensus and the synchronous state are equivalent.

\begin{definition} {Symmetric State}
\label{SymStateDef}

The symmetric state $(\boldsymbol{\rho}^{*})$ is defined as the state where for each unitary permutation $U_{\pi}$, with $\pi \in G_B \subset S_N$
%
\begin{equation}
    \label{eq:AverageState}
    \begin{gathered}
        \boldsymbol{\rho}^{*}  =  \frac{1}{|G_{\textit{B}}|} \sum_{\pi \in \textit{B}} {U_{\pi} \boldsymbol{\rho}(t) U_{\pi}^{\dagger} }
    \end{gathered}
\end{equation}
\end{definition}

For the special case that the subset $B$ is able to generate the permutation set $S_N$ (i.e. $G_B = S_N$) the 
symmetric state is referred to as the consensus state. 

\subsection{Master Equation in Interaction Picture}

For the case of permutation-invariant $H_0 $, i.e.
 \begin{equation}
    \label{eq:Lindbladconstraint334}
    \begin{gathered}
       \left[H_{0},U_{\pi}\right] = 0 \quad \text{for every}\quad \pi \in S_N,
    \end{gathered}
\end{equation}
we can eliminate the first term in Lindblad equation by writing it in interaction picture, i.e. 
$\boldsymbol{\rho}(t)=e^{-iH_0t}\boldsymbol{\rho}_I(t)e^{iH_0t}$ and then substituting the result in Lindblad equation which in turn results in the following,
%
%
\begin{equation}
    \label{eq:Lindblad3}
    \begin{gathered}
        \frac{d\boldsymbol{\rho}_I}{dt} = \sum_{ \pi \in \textit{B}} { w_{\pi} \left(  U_{\pi} \times \boldsymbol{\rho}_I \times U_{\pi}^{\dagger} - \boldsymbol{\rho}_I  \right) }.
     \end{gathered}
\end{equation}
In \cite{PetersenRef15,Petersen2015IEEETranAutControl}, it is shown that equation (\ref{eq:Lindblad3}) can asymptotically reach the symmetric state given below,
\begin{equation}
    \label{eq:QCMEFinalConsensus}
    \begin{gathered}
        \boldsymbol{\rho}^{*}_I(t)  =  \frac{1}{|G_B|} \sum_{\pi \in B} {U_{\pi} \boldsymbol{\rho}_I(t) U_{\pi}^{\dagger} }.
     \end{gathered}
\end{equation}
%
Substituting (8) in (7) and using the fact that $\boldsymbol{\rho}^{*}_I(t)$ is permutation-invariant i.e.
\begin{equation}
    \label{eq:Formula382}
    \begin{gathered}
      U_{\pi} \boldsymbol{\rho}_I^{*}(t)U_{\pi}^{\dagger}=\boldsymbol{\rho}_I^{*}(t)  \quad \text{for every}\quad \pi \in \textit{B},
    \end{gathered}
\end{equation}
it can be concluded that $\frac{d\boldsymbol{\rho}_I^{*}(t)}{dt}=0$ and $\boldsymbol{\rho}_I^{*}(t) = \boldsymbol{\rho}_I^{*}(0)$.
\begin{theorem}
\label{theorem1}
For permutation-invariant $H_0$, the equation (\ref{eq:Lindblad2}) can reach the symmetric state  (\ref{eq:AverageState}).
\end{theorem}

\begin{proof}
%
The following can be written for the 
symmetric state (\ref{eq:AverageState}) 
\begin{equation}
    \label{eq:Formula396}
    \begin{gathered}
        \boldsymbol{\rho}^{*}(t)  =  \frac{1}{|G_{\textit{B}}|} \sum\limits_{\pi \in \textit{B}} {U_{\pi} \boldsymbol{\rho}(t) U_{\pi}^{\dagger} }=
        \frac{1}{|G_{\textit{B}}|} \sum\limits_{\pi \in \textit{B}} {U_{\pi} e^{-iH_0t}\boldsymbol{\rho}_I(t)e^{iH_0t} U_{\pi}^{\dagger}  } 
        = e^{-iH_0t}\boldsymbol{\rho}^*_I(t)e^{iH_0t}=e^{-iH_0t}\boldsymbol{\rho}^*_I(0)e^{iH_0t}. 
     \end{gathered}
\end{equation}
Therefore,  the equation (\ref{eq:Lindblad2}) will reach the symmetric state  (\ref{eq:AverageState}),
as long as the equation (\ref{eq:Lindblad3})   reaches the symmetric state (\ref{eq:QCMEFinalConsensus}).
%

\end{proof}

\begin{theorem}
\label{theorem2}
Equation (\ref{eq:Lindblad2}) can reach the synchronous state if the underlying graph corresponding
to $G_B$ is strongly connected and $H_0$ is permutation-invariant.
\end{theorem}

\begin{proof}
Based on (\ref{eq:Formula396}) we have
\begin{equation}
    \label{eq:Formula404}
    \begin{gathered}
      <i\mid   \boldsymbol{\rho}_l^{*}(t)\mid j>  = tr(\boldsymbol{\rho}^{*}(t)I\times I\times \cdots \overbrace{\mid j> <i\mid }^l\times \cdots I)
      = tr(\boldsymbol{\rho}^{*}(t)U_{\pi}(I\times I\times \cdots \overbrace{\mid j> <i\mid }^k\times \cdots I)U_{\pi}^{\dagger})
      \\
      = tr(U_{\pi}^{\dagger}\boldsymbol{\rho}^{*}(t)U_{\pi}I\times I\times \cdots \overbrace{\mid j> <i\mid }^k\times \cdots I)
      = tr(\boldsymbol{\rho}^{*}(t)I\times I\times \cdots \overbrace{\mid j> <i\mid }^k\times \cdots I)=<i\mid   \boldsymbol{\rho}_k^{*}(t)\mid j>,
     \end{gathered}
\end{equation}
where $U_{\pi}  $ is the unitary representation of a permutation that maps $0\cdots0\overbrace{1}^l0\cdots 0$  to  $0\cdots0\overbrace{1}^k0\cdots 0$.
From equation (\ref{eq:Formula404}), it can be concluded that the $\boldsymbol{\rho}_{l}^{*}(t) = \boldsymbol{\rho}_{k}^{*}(t)$ i.e. the reduced states of the $l$-th and $k$-th qudits in the equilibrium are the same, i.e synchronous state is obtained.

\end{proof}

Also equation(\ref{eq:Formula404}) implies the following,
\begin{equation}
    \label{eq:Formula413}
    \begin{gathered}
        tr(\boldsymbol{\rho}^{*}(t)\sigma_l)  =  tr(\boldsymbol{\rho}^{*}(t)\sigma_k)\quad \text{for}\quad k,l=1,\cdots,N,
    \end{gathered}
\end{equation}
where $\sigma_k= I\otimes I\otimes \cdots \overbrace{\sigma}^k\otimes \cdots I)$.
In \cite{PetersenRef15}, equation $(\ref{eq:Formula413})$ is defined as the $\sigma$-Expectation Consensus.

\begin{remark}
Based on Theorem \ref{theorem2} and Definition \ref{SymStateDef}, we can state that for permutation-invariant $H_0$ and a strongly connected underlying graph $(B)$ which can generate the permutation group (i.e. $G_B = S_N$),
the consensus state is reachable for $\boldsymbol{\rho}_{I}(t)$,
and on the other hand the consensus state is reachable for $\boldsymbol{\rho}(t)$ if $H_0 = 0$. 
\end{remark}

Thus to have both consensus and synchronous states as feasible states, for the rest of the analysis presented in this paper we have assumed that $H_0 = 0$ and $G_B = S_N$.



\subsection{Equivalent Classical Continuous-Time Consensus Algorithm}


The density matrix $(\boldsymbol{\rho})$ can be expanded in terms of the generalized Gell-Mann matrices 
{\cite[Appendix A]{SaberQConsensusContinuous}} as below,
%
\begin{equation}
    \label{eq:DecompositionDensityGeneral}
    \begin{gathered}
        \boldsymbol{\rho} = \frac{1}{d^N} \sum_{ \mu_{1}, \mu_{2}, \ldots, \mu_{N} = 0 }^{ d^{2} - 1 }  {  \rho_{\mu_{1}, \mu_{2}, \ldots, \mu_{N} } \cdot \boldsymbol{\lambda}_{\mu_{1}} \otimes  \lambda_{\mu_{2}} \otimes \cdots \otimes \lambda_{\mu_{N}}  },
     \end{gathered}
\end{equation}
where $N$ is the number of particles in the network and $\otimes$ denotes the Cartesian product and $\lambda$ matrices are the generalized Gell-Mann matrices.
Note that due to Hermity of density matrix, its coefficients of expansion $\rho_{ \mu_{1}, \mu_{2}, \ldots, \mu_{N} }$ are real numbers and because of unit trace of $\boldsymbol{\rho}$ we have $\rho_{0,0,\ldots,0} = 1$.

Substituting the density matrix $\boldsymbol{\rho}$ from (\ref{eq:DecompositionDensityGeneral})
in Lindblad master equation (\ref{eq:Lindblad3}) and considering the independence of the matrices $\lambda_{\mu_{1}} \otimes \lambda_{\mu_{2}} \otimes \cdots \lambda_{\mu_{N}} $ we can conclude the following for Lindblad master equation (\ref{eq:Lindblad3}),
\begin{equation}
    \label{eq:DensityEquation1}
    \begin{gathered}
        \frac{d}{dt} \rho_{\mu_{1}, \cdots, 
        \mu_{N}}  =
        \sum_{\pi \in \textit{B} }   {  w_{\pi} \left( \rho_{\pi(\mu_{1}),\cdots,\pi(\mu_{N}) }
        -
        \rho_{\pi(\mu_{1},\cdots,\mu_{N}) }  \right)  } 
    \end{gathered}
\end{equation}
for all $\mu_{1},\mu_{2},\cdots,\mu_{N}=0,\cdots,d^{2}-1$,
with the constraint (\ref{eq:Lindbladconstraint}) on the edge weights.
The tensor component of the quantum consensus state (\ref{eq:QCMEFinalConsensus}) can be written as below
\begin{equation}
    \label{eq:QuantumConsensusState872}
    \begin{gathered}
        \rho_{\mu_1, \mu_2, \ldots, \mu_N}^{*}  =  \frac{1}{N!} \sum\nolimits_{\pi \in S_N} {\rho_{\pi(\mu_1), \pi(\mu_2), \ldots, \pi(\mu_N)} (0)}
    \end{gathered}
\end{equation}
and for the strongly connected underlying graph, the QCME reaches quantum consensus, componentwise as below
\begin{equation}
    \nonumber
    \begin{gathered}
        \lim_{t \rightarrow \infty} {    \rho_{\mu_1, \mu_2, \ldots, \mu_N}  (t)    }  =  \rho_{\mu_1, \mu_2, \ldots, \mu_N}^{*},
    \end{gathered}
\end{equation}
Comparing the set of equations in (\ref{eq:DensityEquation1}) with those of the classical Continuous-Time Consensus (CTC) problem \cite{SaberQConsensusContinuous} we can
see that the Quantum Consensus Master Equation (\ref{eq:Lindblad3}) is transformed into a classical CTC problem  
with $d^{2N} -1$ tensor component $\rho_{\mu_{1}, \cdots, \mu_{N}}$ as the agents' states.
Defining $\boldsymbol{X}_Q$ as a column vector of length $d^{2N}$ with components $\rho_{\mu_1, \ldots, \mu_N}$, the state update equation of the classical CTC can be written as below,
\begin{equation}
    \label{eq:QuantumStateUpdate}
    \begin{gathered}
        \frac{d\boldsymbol{X}_Q}{dt} = - \boldsymbol{L}_Q \boldsymbol{X}_Q,
    \end{gathered}
\end{equation}
with the constraint (\ref{eq:Lindbladconstraint}) on the edge weights.
$\boldsymbol{L}_Q$ is the corresponding Laplacian matrix as below,
\begin{equation}
    \label{eq:QuantumLaplacian}
    \begin{gathered}
        \boldsymbol{L}_Q  =  \sum\nolimits_{\pi \in \textit{B}} {  \boldsymbol{w}_{\pi} ( I_{d^{2N}}  -  U_{\pi} )  },
    \end{gathered}
\end{equation}
where $U_{\pi}$ is the swapping operator given in 
{\cite[Appendix A]{SaberQConsensusContinuous}}
 provided that $d$ is replaced with $d^2$ which in turn results in Gell-Mann matrices of size  $d^2 \times d^2$.

\subsection{Convergence Rate to Consensus State}

It is obvious that the convergence rate of the obtained classical CTC problem (\ref{eq:QuantumStateUpdate}) is dictated by $Re(\lambda_2 ( \boldsymbol{L}_Q ))$, where $\lambda_2 ( \boldsymbol{L}_Q )$ is the eigenvalue of $\boldsymbol{L}_{Q}$ with second smallest real value.
%
Thus, the corresponding optimization problem can be written as below,
\begin{equation}
    \label{eq:FCTQCInitial}
    \begin{aligned}
        \max\limits_{\boldsymbol{w}} \quad &Re\left( \lambda_2(\boldsymbol{L}_{Q}) \right) \\ 
        s.t. \quad &\sum\nolimits_{\pi \in \textit{B}} {l_{\pi} w_{\pi}} \leq D.
    \end{aligned}
\end{equation}
This problem is known as the Fastest Continuous Time Quantum Consensus (FCTQC) problem \cite{Petersen2015IEEETranAutControl,SaberQConsensusContinuous,Petersen2015ACCPartI}.
In \cite{SaberQConsensusContinuous,Petersen2015IEEETranAutControl} it is shown that the underlying graph of the obtained classical CTC problem (\ref{eq:QuantumStateUpdate}) is a cluster of connected components.
Similar result can be deduced for directed underlying graphs and it can be shown that each strongly connected graph component corresponds to a given partition of $N$ into $K$ integers, namely $N = n_{1} + n_{2} + \cdots + n_{K}$, where $K \leq d^2$ and $n_{j}$ for $j=1,\ldots,K$  is the number of indices in $\rho_{\mu_1, \mu_2, \ldots, \mu_N}$ with equal values. 
For a given partition and its associated Young Tabloids, more than one connected component can be obtained. 
Therefore for each partition we consider only one of them, and we refer to this graph as the induced graph.
These induced graphs are the same as those noted in \cite{Petersen2015IEEETranAutControl}.
Each Young tabloid $t_{n}(r_{1}, r_{2}, \cdots, r_{N})$ is equivalent to an agent in the 
induced graph of the CTC problem
and its corresponding coefficient ($\rho_{\mu_{r_{1}}, \mu_{r_{2}}, \cdots, \mu_{r_{N}}}$) is equivalent to the state of that agent.
For more details on Young Tabloids and their association with partitions of an integer, we refer the reader to \cite{SaberQConsensusContinuous}.
%
%
%
%
%
%
Based on the fact that the underlying graph of the CTC problem (\ref{eq:DensityEquation1}) is a cluster of connected components, in \cite{SaberQConsensusContinuous} it is shown that the Laplacian matrix $\boldsymbol{L}_Q$ is a block diagonal matrix where each block corresponds to one of the connected components, with  state vector  $\boldsymbol{X}_{n}$.
The state update equation (\ref{eq:QuantumStateUpdate}) for the state vector $\boldsymbol{X}_{n}$ is as below,
\begin{equation}
    \label{eq:StateUpdateEquationXn}
    \begin{gathered}
        \frac{d\boldsymbol{X}_n}{dt} = -\boldsymbol{L}_n \boldsymbol{X}_n,
    \end{gathered}
\end{equation}
with $\boldsymbol{L}_n$ as the Laplacian matrix which is one of the blocks in $\boldsymbol{L}_Q$.
Therefore, the convergence rate of (\ref{eq:Lindblad3}) to the fixed point (\ref{eq:QCMEFinalConsensus}) is equivalent to the convergence rate of the obtained classical CTC problem,
i.e. $\left( Re(\lambda_2(\boldsymbol{L}_Q)) \right)$, which in turn is determined by the real part of
the spectrum of the induced graphs' Laplacian matrices.
In other words we have $Re \left( \lambda_2(\boldsymbol{L}_Q) \right) = \min\limits_{n} \left( Re \left( \lambda_2(\boldsymbol{L}_n) \right) \right) $.

Let $n$ and $n^{'}$ be two given partitions of $N$, then $n$ dominates $n^{'}$ if we have
\begin{equation}
    \nonumber
    \begin{aligned}
        n \unrhd n^{'} \quad \text{if and only if} \quad \sum_{j=1}^{i} {n_{j}} \geq \sum_{j=1}^{i} {n_{j}^{'}}  \quad \text{for all} \quad i \geq 1.
    \end{aligned}
\end{equation}
In \cite{SaberQConsensusContinuous}, it is shown that the spectrum of the induced graph corresponding to the dominant partition is included in that of the less dominant partition.
This is known as the intertwining relation \cite{SaberQConsensusContinuous}.
Also, in \cite{SaberQConsensusContinuous}, it is shown that if the underlying graph of the quantum network is a connected and undirected graph then the second smallest eigenvalues of all induced graphs are equal.
This is known as the generalization of Aldous' conjecture.
Using this result, in \cite{SaberQConsensusContinuous}, the problem of optimizing the convergence rate of quantum consensus algorithm is reduced to the problem of maximizing the second smallest eigenvalue of the underlying graph of the quantum network.
Thus, the convergence rate is independent of $d$ (the dimension of the Hilbert space).
In general, Aldous' conjecture does not hold true for the case of directed underlying graphs.
Therefore for optimizing $Re \left( \lambda_2(\boldsymbol{L}_Q) \right)$, all induced graphs should be considered.
Hence, the convergence rate will depend on $d$ (the dimension of the Hilbert space).
In next section, for several examples, we have shown that the Aldous' conjecture is partially true.

\subsection{Convergence Rate to the Synchronous State}

Expanding $\rho^{*}$ in terms of Gell-Mann matrices (similar to (\ref{eq:DecompositionDensityGeneral})) and substituting $\lambda_{\mu}$ as $\sigma$ in (\ref{eq:Formula413}) 
we have
%
\begin{equation}
    \nonumber
    \begin{gathered}
        \rho^{*}_{0, \ldots, 0, \underbrace{\mu}_{l}, 0, \ldots 0} = \rho^{*}_{0, \ldots, 0, \underbrace{\mu}_{k}, 0, \ldots 0}.
    \end{gathered}
\end{equation}
Since (\ref{eq:Formula413}) can be concluded if the equation above holds true therefore (\ref{eq:Formula413}) is the necessary and sufficient condition for reaching the synchronous state. 
In other words, reaching consensus in the underlying graph (not necessarily in the induced graphs) is the necessary and sufficient condition for reaching the synchronous state.

Therefore, it can be concluded that for analyzing the convergence rate to the synchronous state, it suffices to study the convergence rate of the classical consensus algorithm over only the underlying digraph of the network.
As a result the convergence rate to the synchronous state wold be independent of $d$ (the dimension of Hilbert space).
For optimizing the convergence rate to the synchronous state, the corresponding optimization problem can be written as below,
\begin{equation}
    \label{eq:SynchronousOptimizationState}
    \begin{aligned}
        \max\limits_{\boldsymbol{w}} \quad &Re\left( \lambda_2(\boldsymbol{L}_{U}) \right)  \\  
        s.t. \quad &\sum\nolimits_{\pi \in \textit{B}} {l_{\pi} w_{\pi}} \leq D.
    \end{aligned}
\end{equation}
where $\boldsymbol{L}_{U}$ is the Laplacian matrix of the underlying graph.
Note that as long as the Aldous’ conjecture holds true, the convergence rate to both the synchronous and the consensus states are the same, otherwise the convergence rate to the synchronous state is faster than that of the consensus state.


For the rest of this paper, we will study different underlying digraphs and we will investigate if the Aldous' conjecture holds true and also we will study the convergence rate to both the synchronous state and the consensus states.


\section{Optimizing the Convergence Rates to the Consensus and Synchronous States}
\label{Sec:Simulations}

In this section, we optimize the convergence rates of the distributed quantum consensus algorithm to the consensus and the synchronous states, over different topologies with three and four qudits.

Topologies with $3$ qudits have two induced graphs of sizes $3$ and $6$, where the smaller induced graph is identical to the underlying graph of the topology.
we denote the Laplacian matrices of the underlying graph and the induced graphs by $\boldsymbol{L}_U$, $\boldsymbol{L}_3$ and $\boldsymbol{L}_6$, respectively.
Using the intertwining relation \cite{SaberQConsensusContinuous} between laplacian matrices of the induced graphs, we can state that all eigenvalues of $\boldsymbol{L}_3$ are amongst those of $\boldsymbol{L}_6$.
Therefore, the convergence rates to the synchronous and consensus states are dictated by the second smallest eigenvalues of $\boldsymbol{L}_{U}$ and $\boldsymbol{L}_{6}$, respectively.

The first topology $\left( G_{1}(3) \right)$ that we consider is a digraph with three qudits with one cycle of length three and one transposition.
The underlying and the induced $(\boldsymbol{L}_6)$ graphs of this topology are depicted in figure \ref{fig:G3D1Graph}.
The weight on edges of the cycle and the transposition are denoted by $w_{123}$ and $w_{12}$, respectively.
%
%
\begin{figure}
  \centering
     \includegraphics[width=110mm]{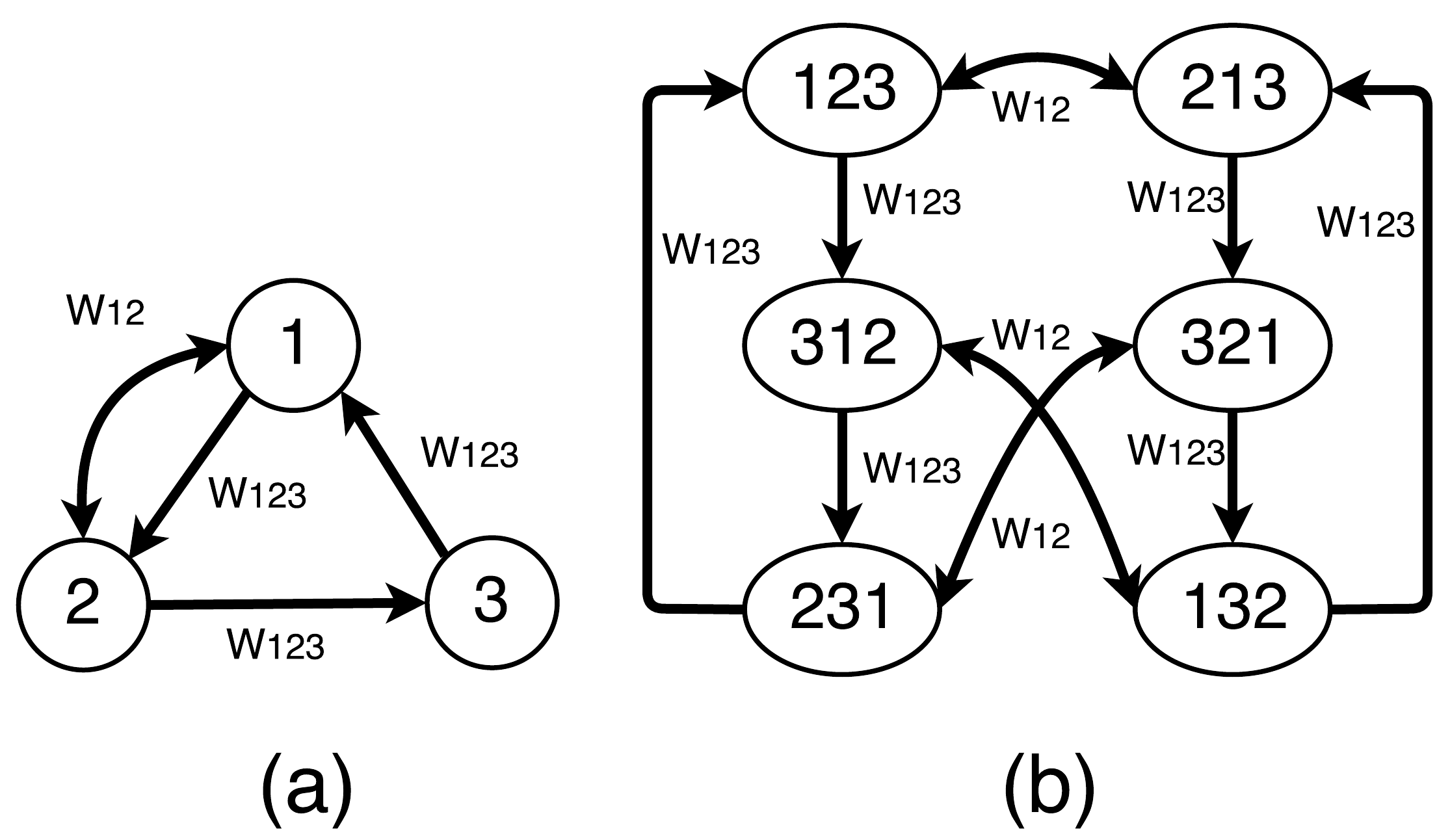}
  \caption{(a) Underlying graph of $G_{1}(3)$ and (b) the resultant induced graph.}
  \label{fig:G3D1Graph}
\end{figure}
%
The Laplacian matrices for the underlying $(\boldsymbol{L}_U)$ and the induced graphs $(\boldsymbol{L}_{6})$ of this topology are as below,
%
\begin{equation}
\begin{gathered}
    \nonumber
    \boldsymbol{L}_U =
    \small{\left[ \begin{array}{ccc}
    { w_{123} + w_{12} }        &{ -w_{12} }                &{ -w_{123} }       \\
    { -w_{123} - w_{12} }       &{ w_{123} + w_{12} }       &{ 0 }              \\
    { 0 }                       &{ -w_{123} }               &{ w_{123} }    \end{array} \right]}, 
\end{gathered}
\end{equation}
%
%
\begin{equation}
\begin{aligned}
    \nonumber
    \boldsymbol{L}_{6} =
    \left[ \begin{array}{cccccc}
    { d_1 }   &{ 0 }   &{ -w_{123} }     &{ -w_{12} }      &{ 0 }      &{ 0 }        \\
    { -w_{123} }   &{ d_1 }   &{ 0 }     &{ 0 }      &{ 0 }      &{ -w_{12} }        \\
    { 0 }   &{ -w_{123} }   &{ d_1 }     &{ 0 }      &{ -w_{12} }      &{ 0 }        \\
    { -w_{12} }   &{ 0 }   &{ 0 }     &{ d_1 }      &{ 0 }      &{ -w_{123} }        \\
    { 0 }   &{ 0 }   &{ -w_{12} }     &{ -w_{123} }      &{ d_1 }      &{ 0 }        \\
    { 0 }   &{ -w_{12} }   &{ 0 }     &{ 0 }      &{ -w_{123} }      &{ d_1 }
    \end{array} \right],
\end{aligned}
\end{equation}
where $d_1 = w_{12} + w_{123}$.
%
%
%
%
The nontrivial eigenvalues of $\boldsymbol{L}_{U}$ are 
$A_1 \pm \sqrt{ B_1 }/2$, where $A_1 = \frac{3}{2}w_{123} + w_{12}$ and $B_1 = -3w_{123}^2 + 4w_{12}^2$.
Thus the second smallest eigenvalue of $\boldsymbol{L}_{U}$ is 
$A_1 - \frac{1}{2} \sqrt{ B_1 }$.
In the case of $\boldsymbol{L}_{6}$ the nontrivial eigenvalues are 
$A_1 - \frac{1}{2} \sqrt{ B_1 }$ and $2w_{12}$
and the second smallest eigenvalue of $\boldsymbol{L}_{6}$ is 
$\min \left( 2w_{12} , A_1 - \frac{1}{2} \sqrt{ B_1 } \right)$.
Hence the convergence rate to the synchronous state and the consensus state is dictated by the following,
%
\begin{equation}
    \begin{gathered}
        \label{eq:G3LambdaSynch}
        \lambda_{Synch} = Re \left( A_1 - \sqrt{ B_1 } / 2 \right),
    \end{gathered}
\end{equation}
%
%
\begin{equation}
    \begin{gathered}
        \label{eq:G3LambdaConsensus}
        \lambda_{Cons} = \min \left( 2w_{12} , Re \left( A_1 - \sqrt{ B_1 } / 2 \right) \right).
    \end{gathered}
\end{equation}
In figure \ref{fig:G3D1Pareto}, we have plotted the Parieto region \cite{BoydBook2004} for $\lambda_{Synch}$ and $\lambda_{Cons}$, with the constraint $3w_{123} + 2w_{12} \leq 1$.
From figure \ref{fig:G3D1Pareto}, it is obvious that there is not any global optimal point for both $\lambda_{Cons}$ and $\lambda_{Synch}$.
The region is bounded by two lines, namely $\lambda_{Synch} = 0.5$ and $\lambda_{Cons} = \lambda_{Synch}$ and a convex curve, where all three of them are obtained for the constraint $3w_{123} + 2w_{12} = 1$.
The line $\lambda_{Synch} = 0.5$ is obtained for $0 < w_{123} \leq \frac{1}{5} $ where 
$B_1 \leq 0$
and thus 
$Re \left( A_1 - \sqrt{ B_1 } / 2 \right)  = \frac{3}{2}w_{123} + w_{12} = \frac{1}{2}$.
The line $\lambda_{Cons} = \lambda_{Synch}$ is obtained for $\frac{1}{3+\sqrt{3}} \leq w_{123} \leq \frac{1}{3}$ which results in 
$2w_{12} > Re \left( A_1 - \sqrt{ B_1 } / 2 \right) $.
The convex curve between two lines is obtained for $\frac{1}{5} \leq w_{123} \leq \frac{1}{3+\sqrt{3}}$.
%
Note that the points along the line $\lambda_{Cons} = \lambda_{Synch}$ are the only set of convergence rates where the Aldous' conjecture holds true.
Regarding the consensus state, the optimal convergence rate is the point $\lambda_{Cons} = 0.4$ and  $\lambda_{Synch} = 0.4$ in the pareto region which is obtained for $w_{123} = w_{12} = 1/5$ and for synchronous state, the optimal convergence rate can be obtained for any of the points along the line $\lambda_{Synch} = 0.5$.
Although the points on the convex curve in figure \ref{fig:G3D1Pareto} do not result in any of the optimal convergence rates but these points act as maximal points which can be used for trade-off between the convergence rates to the consensus states and the synchronous state.
\begin{figure}
  \centering
     \includegraphics[width=150mm]{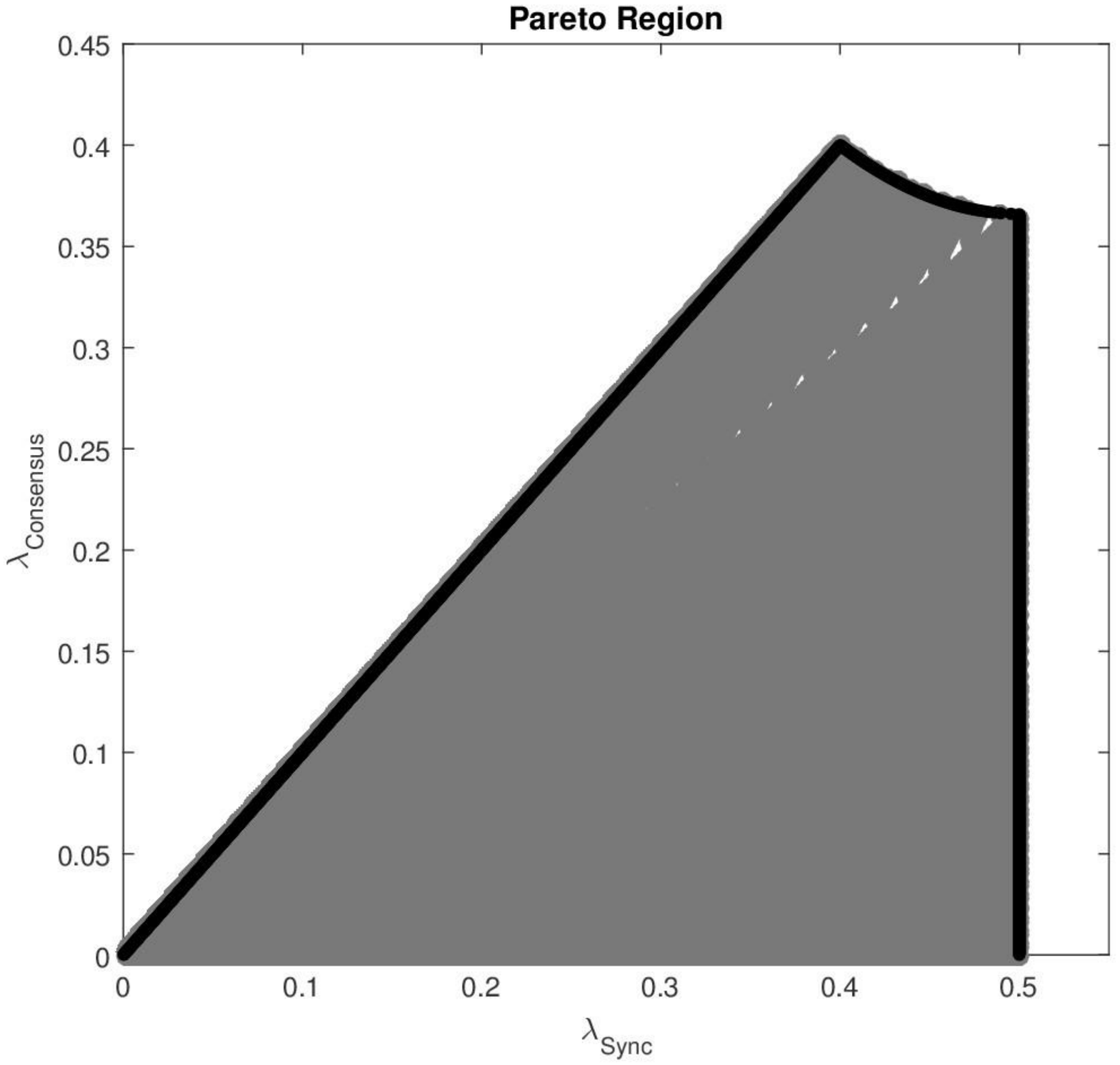}
  \caption{Pareto region for $\lambda_{Cons}$ and  $\lambda_{Synch}$ of the graph $\left( G_{1}(3) \right)$ depicted in figure \ref{fig:G3D1Graph} (a).}
  \label{fig:G3D1Pareto}
\end{figure}

The second topology $\left( G_{2}(3) \right)$ that we consider is a digraph with three qudits, two cycles and one transposition.
The underlying and the induced $(\boldsymbol{L}_{6})$ graphs of this topology are depicted in figure \ref{fig:G3D2Graph}.
The weight on edges of the cycles and the transposition are denoted by $w_{123}$, $w_{321}$ and $w_{12}$, respectively.
%
\begin{figure}
  \centering
     \includegraphics[width=110mm]{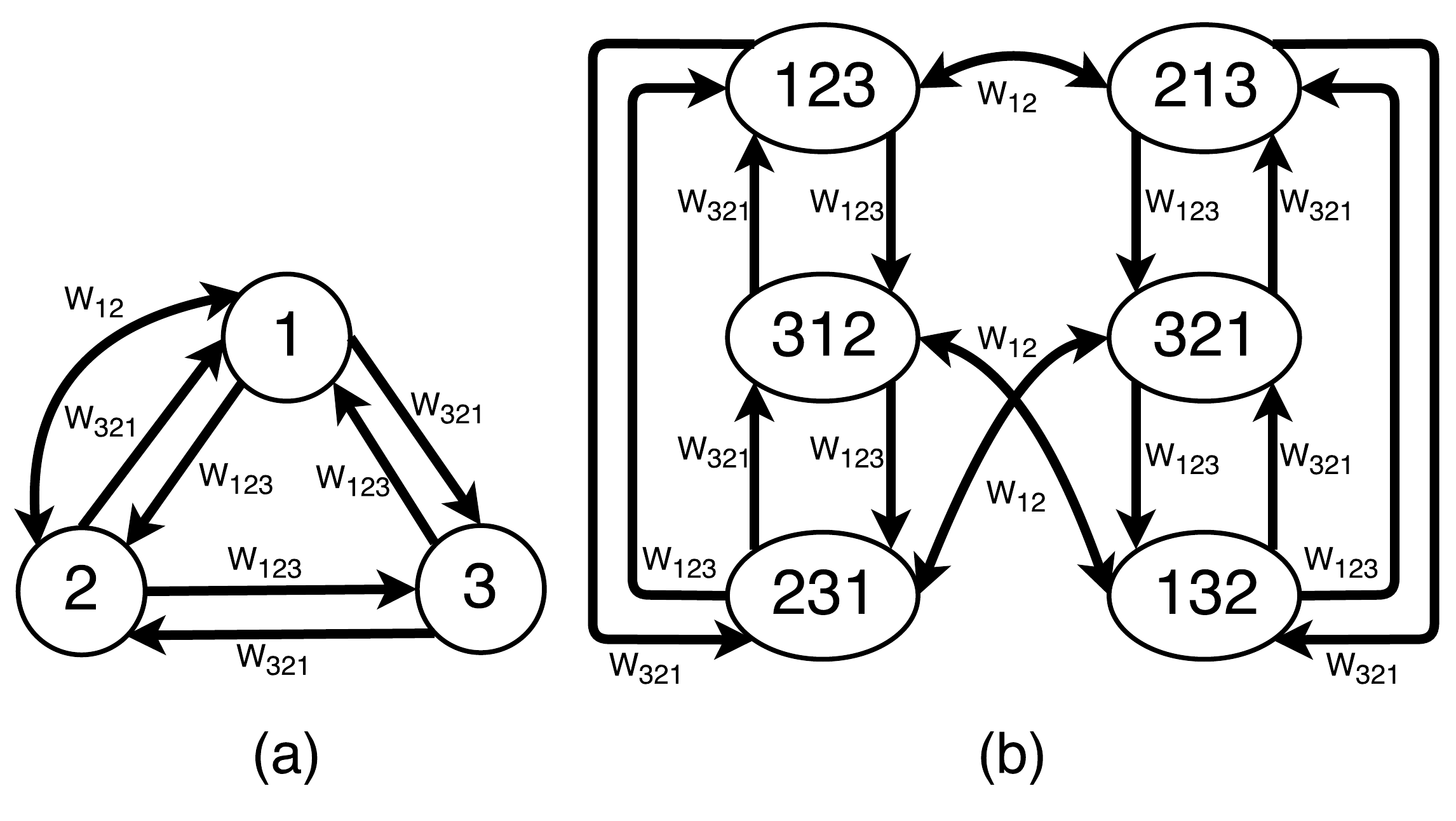}
  \caption{(a) Underlying graph $G_{1}(3)$ and (b) the resultant induced graph.}
  \label{fig:G3D2Graph}
\end{figure}
%
The Laplacian matrices for the underlying $(\boldsymbol{L}_U)$ and the induced $(\boldsymbol{L}_{6})$ graphs of this topology are as below,
\begin{equation}
\begin{gathered}
    \nonumber
    \boldsymbol{L}_3 =
    \small{\left[ \begin{array}{ccc}
    { w_{123} + w_{321} + w_{12} }        &{ -w_{321}-w_{12} }                &{ -w_{123} }       \\
    { -w_{123} - w_{12} }       &{ w_{123} + w_{321} + w_{12} }        &{ -w_{321} }              \\
    { -w_{321} }                       &{ -w_{123} }               &{ w_{123} + w_{321 }}    \end{array} \right]},
\end{gathered}
\end{equation}
%
%
\begin{equation}
\begin{aligned}
    \nonumber
    \boldsymbol{L}_6 =
    \left[ \begin{array}{cccccc}
    { d_2 }   &{ -w_{321} }   &{ -w_{123} }     &{ -w_{12} }      &{ 0 }      &{ 0 }        \\
    { -w_{123} }   &{ d_2 }   &{ -w_{321} }     &{ 0 }      &{ 0 }      &{ -w_{12} }        \\
    { -w_{321} }   &{ -w_{123} }   &{ d_2 }     &{ 0 }      &{ -w_{12} }      &{ 0 }        \\
    { -w_{12} }   &{ 0 }   &{ 0 }     &{ d_2 }      &{ -w_{321} }      &{ -w_{123} }        \\
    { 0 }   &{ 0 }   &{ -w_{12} }     &{ -w_{123} }      &{ d_2 }      &{ -w_{321} }        \\
    { 0 }   &{ -w_{12} }   &{ 0 }     &{ -w_{321} }      &{ -w_{123} }      &{ d_2 }
    \end{array} \right],
\end{aligned}
\end{equation}
where $d_2 = w_{123} + w_{321} + w_{12}$.
%
%
%
%
%
The nontrivial eigenvalues of $\boldsymbol{L}_{U}$ are 
$A_2   \pm   \sqrt{ B_2 } / 2$
where $A_2 = \frac{3}{2} w_{123}   +   \frac{3}{2} w_{321}   +   w_{12}$ and $B_2 = -3w_{321}^{2} + 6 w_{321} w_{123} + 4 w_{12}^2 - 3 w_{123}^{2}$ and the second smallest eigenvalue of $\boldsymbol{L}_{U}$ is $A_2   -   \sqrt{ B_2 }/2$.
In the case of $\boldsymbol{L}_{6}$, the nontrivial eigenvalues are 
$A_2   \pm   \sqrt{ B_2 }/2$
and $2w_{12}$ and the second smallest eigenvalue of $\boldsymbol{L}_{6}$ is 
$\min \left( 2w_{12} , A_2   -   \sqrt{ B_2 }/2 \right)$.
Hence the convergence rate to the synchronous state and the consensus state is dictated by the following,
%
\begin{equation}
    \begin{gathered}
        \label{eq:G3D2LambdaSynch}
        \lambda_{Synch} = Re \left( A_2   -   \sqrt{ B_2 } / 2 \right),
    \end{gathered}
\end{equation}
%
%
\begin{equation}
    \begin{gathered}
        \label{eq:G3D2LambdaConsensus}
        \lambda_{Cons} = \min \left( 2w_{12} , Re \left( A_2   -   \sqrt{ B_2 }/2 \right) \right).
    \end{gathered}
\end{equation}
With constraint $3w_{123} + 3w_{321} + 2w_{12} \leq 1$, format of the pareto region for this topology is similar to that of $G_1(3)$, i.e. the pareto region is bounded by the vertical line $\lambda_{Synch} = 0.5$, the line $\lambda_{Cons} = \lambda_{Synch}$ and a convex curve between these two lines.
The boundaries of the pareto region are obtained for the case that either one of $w_{321}$ or $w_{123}$ is zero.
In this case, this topology reduces to $G_{1}(3)$.
%
In other words, the second cycle is redundant as it is slowing down the convergence rates to both the consensus state and the synchronous state.
Interestingly, for the case $w_{321} = w_{123}$ (where the underlying graph is an undirected graph), the Aldous's conjecture does not always hold true.
As an example for $w_{321} = w_{123} > \frac{2}{3} w_{12}$, the convergence rates are $\lambda_{Synch} = 3w_{123} > \lambda_{Cons} = 2w_{12}$.
%
%
%
Similar to $G_{1}(3)$ topology, the points along the line $\lambda_{Cons} = \lambda_{Synch}$ are the only set of convergence rates where the Aldous' conjecture holds true.

The third topology $\left( G_{3}(3) \right)$ that we have analyzed is a digraph with three qudits and two transpositions.
This topology is identical to an undirected path graph with three vertices.
The weight on edges of the transpositions are denoted by $w_{12}$ and $w_{23}$.
%
The Laplacian matrices for the underlying $(L_U)$ of this topology is as below,
\begin{equation}
\begin{gathered}
    \nonumber
    \boldsymbol{L}_U =
    \small{\left[ \begin{array}{ccc}
    { w_{12} }        &{ -w_{12} }                &{ 0 }       \\
    { -w_{12} }       &{ w_{12} + w_{23} }        &{ -w_{23} }              \\
    { 0 }             &{ -w_{23} }                &{ w_{23} }    \end{array} \right]},
\end{gathered}
\end{equation}
In \cite{SaberQConsensusContinuous} it is shown that for topologies with undirected and connected underlying graphs, the convergence rate of the quantum consensus algorithm to the consensus state is obtained from the second smallest eigenvalue of the Laplacian matrix of the underlying graph.
Thus it can be concluded that independent of the value of the weights, for topologies with undirected and connected underlying graphs, the convergence rates to the consensus and the synchronous states are the same.
Hence for topology  $G_{3}(3)$, the convergence rates to consensus and the synchronous states are always equal and
considering the constraint $2w_{12} + 2w_{23} = 1$,
the pareto region for this topology would be a direct line between points $\lambda_{Cons} = \lambda_{Synch} = 0$  and $\lambda_{Cons} = \lambda_{Synch} = 1/4$.
Point $\lambda_{Cons} = \lambda_{Synch} = 1/4$ is the global optimal point in terms of both convergence rates which is obtained for $w_{12} = w_{23} = 1/4$.

The fourth topology that we have considered in this section is a digraph with four qudits, one cycle of length four and two transpositions.
The Laplacian matrix of this topology is as below,
\begin{equation}
\begin{gathered}
    \nonumber
    \boldsymbol{L}_U =
    \left[ \begin{array}{cccc}
    { w_{1234} + w_{12} }        &{ -w_{12} }                &{ 0 }       &{ -w_{1234} }      \\
    { -w_{1234} - w_{12} }       &{ w_{1234} + w_{12}  }      &{ 0 }       &{ 0 }      \\
    { 0 }             &{ -w_{1234} }      &{ w_{1234} + w_{34} }       &{ -w_{34} }      \\
    { 0 }             &{ -w_{23} }                &{ -w_{1234} - w_{34} }       &{ w_{1234} + w_{34} }      \end{array} \right],
\end{gathered}
\end{equation}

The weight on edges of the cycle and the transposition are denoted by $w_{1234}$ and $w_{12}$, respectively.
This topology has four induced graphs of sizes $4$, $6$, $12$ and $24$.
We denote the Laplacian matrices of these induced graphs by $L_{4}$, $L_{6}$, $L_{12}$, $L_{24}$ and the second smallest eigenvalues of each one of these matrices by $\lambda_{2}(L_{4})$, $\lambda_{2}(L_{6})$, $\lambda_{2}(L_{12})$ and $\lambda_{2}(L_{24})$.
$L_{4}$ is identical to the Laplacian matrix of the underlying graph of the topology.
The smallest eigenvalue of each one of these Laplacian matrices is zero.
Using the intertwining relation \cite{SaberQConsensusContinuous} between laplacian matrices of the induced graphs,
we can state that all eigenvalues of $L_{4}$, $L_{6}$ and $L_{12}$ are amongst those of $L_{6}$, $L_{12}$ and $L_{24}$, respectively.
%
%
Since $L_{12}$ includes all irreducible representations of $L_{24}$ (except the one corresponding to the largest eigenvalue) then based on \cite{SaberQConsensusDiscrete}, it can be concluded that all eigenvalues of $L_{24}$ (except its largest eigenvalue) are amongst the eigenvalues of $L_{12}$.
Therefore, the convergence rates to the synchronous and consensus states are equal to $\lambda_{2}(L_{4})$ and
$\lambda_{2}(L_{12})$, respectively.

As depicted in figure \ref{fig:G4D5Pareto}, the boundary of the pareto region for this topology is bounded by lines $\lambda_{Cons} = \lambda_{Synch}$ and $\lambda_{Synch} = 0.25$ and two concave curves.
The line $\lambda_{Cons} = \lambda_{Synch}$ stretches  between points $\lambda_{Cons} = \lambda_{Synch} = 0$ and $\lambda_{Cons} = \lambda_{Synch} = 0.1326$ where the latter point is obtained for $w_{1234} = 0.09745$, $w_{12} = 0.1548$, $w_{34} = 0.1503$ and it is denoted as point $A$ in figure \ref{fig:G4D5Pareto}.
The vertical line $\lambda_{Synch} = 0.25$ stretches  between points $\lambda_{Cons} = 0 $, $\lambda_{Synch} = 0.25$ and $\lambda_{Cons} = 0.1699$, $\lambda_{Synch} = 0.25$ where the latter point is obtained for $w_{1234} = 0.1535$, $w_{12} = 0.097$, $w_{34} = 0.096$ and it is denoted as point $C$ in figure \ref{fig:G4D5Pareto}.
Note that this point is the global optimal point in regards to the convergence rates of the algorithm to both consensus and synchronous states.
The point that two boundary curves 
meet each other
is $\lambda_{Cons} = 0.15457 $, $\lambda_{Synch} = 0.19731$ which is obtained for $w_{1234} = 0.11995$, $w_{12} = 0.212$, $w_{34} = 0.0481$ and it is denoted as point $B$ in figure \ref{fig:G4D5Pareto}.

\begin{figure}
  \centering
     \includegraphics[width=140mm]{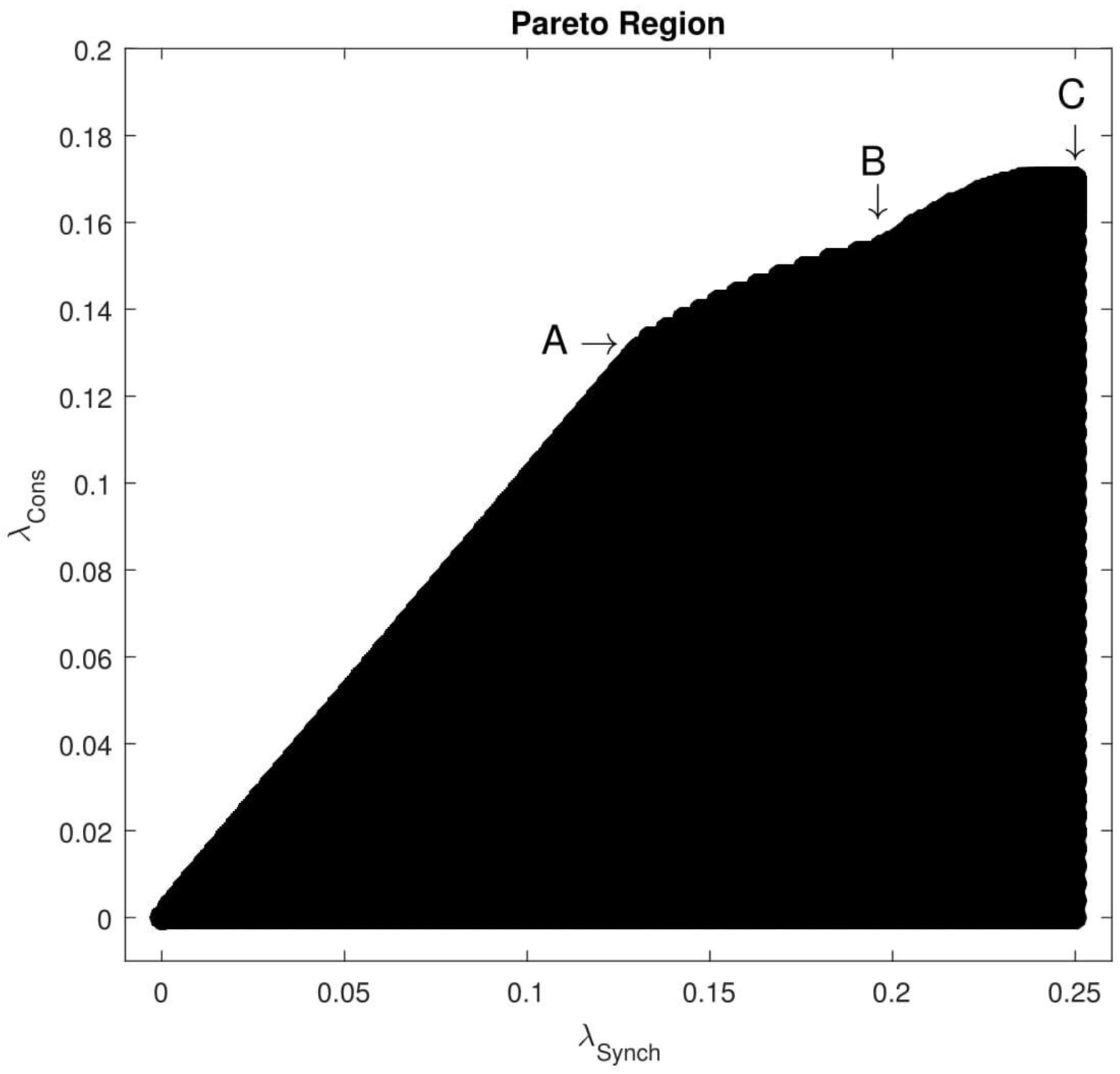}
  \caption{Pareto region for $\lambda_{Cons}$ and  $\lambda_{Synch}$ of the graph $G_{1}(4)$.}
  \label{fig:G4D5Pareto}
\end{figure}

\section{Conclusions}
\label{sec:Conclusion}

Considering a network of qudits, we have studied the convergence rate of the distributed consensus algorithm with general (i.e. either directed or undirected) underlying topology towards consensus and synchronous states.
%
%
We have established that the convergence rate to both states are equal iff the Aldous' conjecture holds true.
In case of networks that their underlying graph contains cycles, the Aldous' conjecture does not necessary hold true and it should be analyzed per case for each combination of weights.
%
%
%
%
%
%
%
%
%
%
%
%
%
%
In our future work, we will study the relation between the convergence rates and the Aldous' conjecture by
relaxing the consensus state to a symmetric state which is invariant to only a subset of all permutations.
%
%
%
Other future studies will focus on analysing the discrete-time model of the quantum consensus algorithm and the quantum gossip algorithm over quantum networks with general underlying topologies.

\appendices


\begin{thebibliography}{99}





%
%





%
%







\bibitem{BoydBook2004}
S.~Boyd and L.~Vandenberghe, \emph{Convex Optimization}, \hskip 1em plus
  0.5em minus 0.4em\relax Cambridge University Press, New York, NY, USA. 2004.




%
\bibitem{Xiao04Boyd}
L.~Xiao and S.~Boyd, \emph{Fast Linear Iterations for Distributed Averaging}, \hskip 1em plus
  0.5em minus 0.4em\relax Systems and Control Letters, vol. 53, pp. 65-78, 2004.



%
%




%
%




\bibitem{ProofAldous}
P.~Caputo and T.~M.~Liggett, \emph{Proof of Aldous’ spectral gap conjecture}, \hskip 1em plus
  0.5em minus 0.4em\relax Journal American Math Society, vol. 23, pp. 831-851, 2010.





\bibitem{Chung1997}
F.~R.~K.~Chung, \emph{Spectral graph theory}, \hskip 1em plus
  0.5em minus 0.4em\relax American Mathematical Society, 1997.



%
%
%



%
%





%
%
%







%
%





%
%






%
%






%
%







%
%














\bibitem{SaberThesis2015}
S.~Jafarizadeh, \emph{Distributed coding and algorithm optimization for large-scale networked systems}, \hskip 1em plus
  0.5em minus 0.4em\relax Ph.D. Dissertation. University of Sydney, NSW, 2015.








%
%
%










%
%
%












%
%













\bibitem{MazzarellaCDC2013}
L.~Mazzarella, A.~Sarlette and F.~Ticozzi, \emph{A new perspective on gossip iterations: From Symmetrization to quantum consensus}, \hskip 1em plus
  0.5em minus 0.4em\relax IEEE 52nd Annual Conference on Decision and Control (CDC), pp. 250-255, 2013.











\bibitem{PetersenRef15}
L.~Mazzarella, A.~Sarlette and F.~Ticozzi, \emph{Consensus for Quantum Networks: Symmetry From Gossip Interactions}, \hskip 1em plus
  0.5em minus 0.4em\relax IEEE Trans. Automat. Control, vol. 60, no. 1, pp. 158-172, Jan. 2015.











%










%










%
%












%














%











\bibitem{Petersen2015IEEETranAutControl}
G.~Shi, D.~Dong, I.~R.~Petersen and K.~Johansson, \emph{Reaching a Quantum Consensus: Master Equations that Generate Symmetrization and Synchronization}, \hskip 1em plus
  0.5em minus 0.4em\relax IEEE Trans. Automat. Control, pp. 1-14, 2015.













\bibitem{Petersen2015ACCPartI}
G.~Shi, S.~Fu and I.~R.~Petersen, \emph{Quantum network reduced-state synchronization part I-convergence under directed interactions}, \hskip 1em plus
  0.5em minus 0.4em\relax American Control Conference (ACC), 2015.










\bibitem{Petersen2015ACCPartII}
G.~Shi, S.~Fu and I.~R.~Petersen, \emph{Quantum network reduced-state synchronization part II-the missing symmetry and switching interactions}, \hskip 1em plus
  0.5em minus 0.4em\relax American Control Conference (ACC), pp. 92-97, 2015.











%




%







\bibitem{SaberQConsensusContinuous}
S.~Jafarizadeh, \emph{Optimizing the Convergence Rate of the Continuous Time Quantum Consensus}, \hskip 1em plus
  0.5em minus 0.4em\relax arXiv 1509.05823, 2015.







\bibitem{SaberQConsensusDiscrete}
S.~Jafarizadeh, \emph{Optimizing the Convergence Rate of the Quantum Consensus: A Discrete Time Model}, \hskip 1em plus
  0.5em minus 0.4em\relax arXiv 1510.05178v2, 2015.



\end{thebibliography}
\end{document}